\documentclass[10pt,twoside]{article}

\usepackage{7OSME-style}

\let\displaystyle=\textstyle

\newtheorem{lemma}[theorem]{Lemma}

\usepackage{wrapfig}
\usepackage{array}

\let\epsilon=\varepsilon
\def\vechat#1{\mathbf{\hat{#1}}}

\def\R{{\mathbb R}}

\title{Conic Crease Patterns with Reflecting Rule Lines}
\authorlist[Demaine, Demaine, Huffman, Koschitz, Tachi]{Erik D. Demaine, Martin L. Demaine, David A. Huffman, Duks Koschitz, Tomohiro Tachi} 
\affiliations{%
Erik D. Demaine \at Computer Science and Artificial Intelligence Laboratory, Massachusetts Institute of Technology,  32 Vassar Street, Cambridge, MA 02139, USA, \email{edemaine@mit.edu} 
\and Martin L. Demaine \at Computer Science and Artificial Intelligence Laboratory, Massachusetts Institute of Technology,  32 Vassar Street, Cambridge, MA 02139, USA, \email{mdemaine@mit.edu} 
\and David A. Huffman \at Department of Computer Science, University of California,
  Santa Cruz, CA 95064, USA
\and Duks Koschitz \at School of Architecture, Pratt Institute,  200 Willoughby Avenue Brooklyn, NY 11205, USA, \email{duks@pratt.edu} 
\and Tomohiro Tachi \at Graduate School of Arts and Sciences, The University of Toyo, Komaba 3-8-1, Meguro-ku, Tokyo 153-8902, JAPAN, \email{tachi@idea.c.u-tokyo.ac.jp}
}

\makeatletter
  \let\runtitle\@title
  \let\runauthor\shortauthor
\makeatother

\begin{document}

\maketitle

\begin{abstract}
  We characterize when two conic curved creases are compatible with each other,
  when the rule lines must converge to conic foci and reflect at the crease.
  Namely, two conics are compatible (can be connected by rule segments in a
  foldable curved crease pattern) if and only if they have equal or reciprocal
  eccentricity.  Thus, circles (eccentricity~$0$) and parabolas
  (eccentricity~$1$) are compatible with only themselves (when scaled
  from a focus), and ellipses (eccentricity strictly between $0$ and $1$)
  and hyperbolas (eccentricity above $1$) are compatible with themselves
  and each other (but only in specific pairings).
  The foundation of this result is a general condition relating
  any two curved creases connected by rule segments.
  We also use our characterization to analyze several curved crease designs.
\end{abstract}

\section{Introduction}
\label{sec:introduction}
Curved folding has attracted artists, designers, engineers, scientists and mathematicians~\cite{Sternberg2009, CurvedCrease_Symmetry},
but its mathematics and algorithms remain major challenges in origami science.  
The goal of this work is to characterize curved-crease
origami possible within a particular restricted family of designs, roughly
corresponding to the extensive curved-crease designs of the third author
\cite{Huffman_Origami5, HuffmanLens_Origami6, DuksThesis}.
Specifically, we assume three properties of the design,
together called \emph{naturally ruled conic curved creases}:
\begin{enumerate}
\item Every crease is curved and a \emph{quadratic spline}, i.e., decomposes
into pieces of conic sections (circles, ellipses, parabolas, or hyperbolas).
\item The \emph{rule segments} (straight line segments on the 3D folded surface)
within each face of the crease pattern \emph{converge} to a common point (i.e.,
pass through that point if extended to infinite lines called \emph{rule lines}),
specifically, a \emph{focus} of an incident conic.
(As in the projective plane, we view parabolas as having one focus at infinity,
and lines as having two identical foci at infinity; rule lines meeting at
a point at infinity means that they are all parallel.)
As a result, the folded state is composed of (general) cones and cylinders.
\item Each crease has a \emph{constant fold angle} all along its length.
By an equation of \cite{Fuchs-Tabachnikov-2007-developable},
this constraint is equivalent to the rule segments \emph{reflecting} through
creases, i.e., whenever a rule segment touches the interior of a crease,
its reflection through the crease is also (locally) a rule segment (effectively
forming a locally flat-foldable vertex at the crease).%
\footnote{In his early work, the third author called this ``refraction'', by
  analogy to optics \cite{DuksThesis}, but geometrically it is reflection.}
\end{enumerate}

This family of curved-crease origami designs is natural because, if
rule segments converge to a focus of a conic, then the reflected rule segments
on the other side of the conic also converge to a focus of that conic.
In this way, conics provide a relatively easy way to bridge between pencils%
\footnote{We use the term ``pencil'' from projective geometry to refer to
  infinite families of lines or segments that converge to a common point.}
of rule segments that converge to various points.

We show in Section~\ref{sec:conics}
that designs within this family must, rather surprisingly, satisfy a
stringent constraint: any two conics that \emph{interact} in the sense of
being connected by rule segments must have
either identical or reciprocal \emph{eccentricity}
(the constant ratio between, for every point on the conic, its distance to a
focus point and its distance to a directrix).
The eccentricity is always nonnegative and finite; $0$ for circles;
$1$ for parabolas; strictly between $0$ and $1$ for ellipses, and
$>1$ for hyperbolas.  Because $0$ has infinite reciprocal and $1$ is
its own reciprocal, circles can interact only with
circles and parabolas can interact only with parabolas.
Ellipses can interact with ellipses of the same eccentricity,
and hyperbolas can interact with hyperbolas of the same eccentricity,
but also every eccentricity of ellipse has exactly one eccentricity of
hyperbola that it can interact with.
Table~\ref{table:compatible} summarizes the possibilities.
Further, we show that interacting conics of the same eccentricity (and
thus the same curve type) are scalings of each other through the shared
focus that the rule segments converge to.

\begin{table}
  \centering
  {\small
  \begin{tabular}{l|m{6mm}m{23mm}m{30mm}m{23mm}}
              & circle & ellipse           & parabola & hyperbola \\
    \hline
    circle    &  yes   &    no             &    no    & no \\
    ellipse   &   no   & yes if scaled     &    no    & yes if reciprocal eccentricity  \\
    parabola  &   no   &    no             &   yes if scaled, shifted, or mirrored    & no \\
    hyperbola &   no   & yes if reciprocal eccentricity &    no    & yes if scaled \\
  \end{tabular}
  }
  \caption{Which conics are compatible in the sense that they are foldable
    when connected by rule segments that converge to a common focus
    and reflect at the conic creases.}
  \label{table:compatible}
\end{table}

The foundation of these results about naturally ruled curved crease patterns
is a general result about any two curved creases
that interact by being connected by rule segments.
Following the formalism introduced at last OSME \cite{HuffmanLens_Origami6},
we give a differential equation defining
a necessary relation between any two such curved creases in Section~\ref{sec:between-creases}.
This condition, in particular, lets you compute the fold angle along one curved
crease from the fold angle along the other curved crease (assuming the given
ruling), which by propagation can determine all creases' fold angles from one
crease's fold angle.  
Combined with a known condition along a single curved crease, these conditions
are in some sense complete, characterizing foldability of curved crease--rule
patterns other than closure constraints; see Theorem~\ref{vertex-free}.
The conditions also give a nice construction of \emph{rigid-ruling}
folding motions (which preserve rulings): they exist provided any single 3D
folded state exists (similar to rigid origami \cite{Tachi-2008-quad}).


Finally, in Section~\ref{sec:models}, we use these results to analyze several
conic-crease designs by the third author.  In some cases, we show that the
natural ruling works well.  In other cases, we show that the natural ruling
cannot possibly work, as it violates our eccentricity constraint.
The latter ``impossible'' designs still fold well in practice; all this means
is that the ruling must be different than than the intended (natural) ruling.

\section{Curved Folding Primer}
\label{sec:primer}
\subsection{Notation}
Our notation follows~\cite{HuffmanLens_Origami6}, but has been somewhat simplified
to focus on the case of interest and to exploit the previously proved
structural properties of curved creases.
In particular, we consider only ``smoothly folded'' ($C^1$)
``curved'' (not straight) creases, which in fact implies that the folded
crease is a $C^2$ nonstraight curve \cite[Corollary~20]{HuffmanLens_Origami6}.
Furthermore, we restrict when every crease is \emph{uniquely ruled},
i.e., every point of the crease has a unique rule segment on either side,
which is equivalent to forbidding flat patches and
cone rulings (where many rule segments share a point of the crease).
If presented with a crease with any of these complications (nonsmooth point,
transition to straight, corner of a flat patch, or apex of a cone ruling),
we can call that crease point a ``vertex'', subdivide the crease at all such
vertices, and then focus on the (curved) subcreases.

By the bisection property \cite[Theorem~8]{HuffmanLens_Origami6}, such nice creases allow the following notation of the signed curvature and a consistent top-side Frenet frame.
Refer to Figure~\ref{fig:bisection}.

\paragraph{2D crease.} For a point $\vec x(s)$ on an arc-length-parameterized $C^2$ 2D crease $\vec x : (0,\ell) \to \mathbb R^2$, we have a top-side Frenet frame $\left(\vec t(s), \vechat n(s), \vechat b(s)\right)$, where $\vec t(s) :={d \vec x(s) \over d s}$ is the tangent vector, $\vechat b(s) := \vec e_z$ is the front-side normal vector of the plane, and $\vechat n(s) := \vechat b(s) \times \vec t(s)$ is the left direction of the crease.
We call $\hat k(s) = {d \vec t(s) \over d s}\cdot \vechat n(s)$ the \emph{signed curvature} of the unfolded crease.

\paragraph{Folded crease.}
For a point $\vec X(s)$ on an arc-length-parameterized $C^2$ folded crease $\vec X : (0,\ell) \to \mathbb R^3$, we have a top-side Frenet frame $\left(\vec T(s), \vechat N(s), \vechat B(s)\right)$, where $\vec T(s) :={d \vec X(s) \over d s}$ is the tangent vector, $\vechat B(s)$ is the top-side normal of the osculating plane of the curve, and $\vechat N(s) := \vechat B(s) \times \vec T(s)$ is the left direction of the crease.
Here, $\vechat B(s)$ consistently forms positive dot products with the surface normals on left and right sides by the bisection property \cite[Theorem~8]{HuffmanLens_Origami6}.
Now, $\hat K(s) = {d \vec T(s) \over d s}\cdot \vechat N(s)$ is called the \emph{signed curvature} of the folded crease, and $ \tau(s) = -{d \vechat B(s) \over d s} \cdot \vechat N(s)$ is called the \emph{torsion}.
If the signed curvature is positive, the curve turns left, and if it is negative, the curve turns right with respect to the front side of the surface.
This top-side Frenet frame satisfies the usual Frenet--Serret formulas:
$$ 
\begin{bmatrix}
     0     & \hat K(s)     & 0       \\
     -\hat K(s) & 0        & \tau(s) \\
     0     & -\tau(s) & 0       \\
   \end{bmatrix}
 \cdot
   \begin{bmatrix}
     \vec T(s) \\
     \vechat N(s) \\
     \vechat B(s) \\
   \end{bmatrix}
 =
   {d \over d s}
   \begin{bmatrix}
     \vec T(s) \\
     \vechat N(s) \\
     \vechat B(s) \\
   \end{bmatrix}.
$$

\begin{figure}
  \centering
  \includegraphics[width=0.8\linewidth]{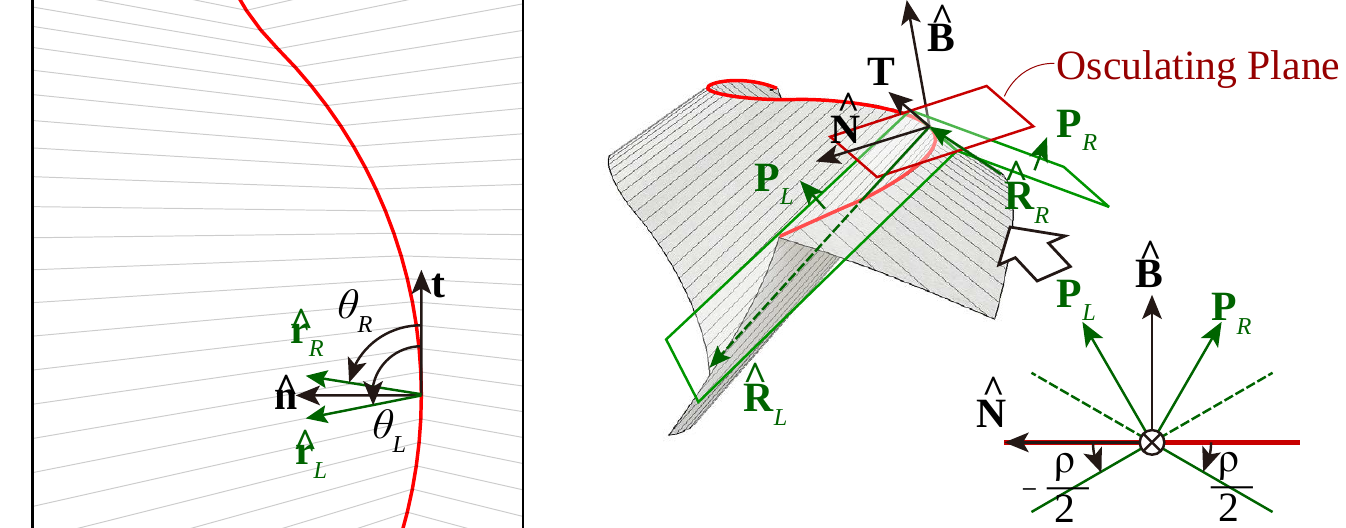}
  \caption{Notation around a folded curved crease.}
  \label{fig:bisection}
\end{figure}

\paragraph{Folded Surfaces Around a Crease.}
The neighborhood of a crease $\vec x$ folds to two developable surfaces attached to~$\vec X$.
The surface normals of the left and right surfaces at crease point $\vec X(s)$ are denoted by $\vec P_L(s)$ and $\vec P_R(s)$.
The bisection property \cite[Theorem~8]{HuffmanLens_Origami6} states that vector $\vechat B(s)$ bisects $\vec P_L(s)$ and $\vec P_R(s)$.
Define the fold angle as the signed angle $\rho$ from $\vec P_R$ to $\vec P_L$ in the right-screw direction of~$\vec T$.
By the bisection property, $\vec P_L$ is a $\frac 1 2 \rho$ rotation of $\vechat B$ and $\vec P_L$ is a $-\frac 1 2 \rho$ rotation of $\vechat B$ around $\vec T$.
By the flat intrinsic isometry of paper, the geodesic (signed) curvature of $\vec X(s)$ on the left and right surfaces must be equal to $\hat k(s)$. Thus we obtain the following basic equation:
\begin{align}
\label{eq:curvature}
\hat K(s)\cos \frac1 2 \rho(s) &= \hat k(s).
\end{align} 

\subsection{Single Crease Condition}
We review the conditions that must hold locally at a curved crease,
as mentioned without derivation e.g. in \ \cite{Fuchs-Tabachnikov-1999},
but adapted to our terminology.
\paragraph{Ruling vectors and angles.}
Define $\vechat R_L(s)$ and $\vechat R_R(s)$ to be the unit ruling vectors of the left and right surfaces (i.e., defining segments from $\vec X(s)$ on these surfaces) but with sign chosen to make them directed to the left side, i.e., $\vechat R_L(s) \cdot \vechat N(s)$ and $\vechat R_R(s) \cdot \vechat N(s)$ are both positive.
Define \emph{signed ruling angles} $\hat \theta_L(s)$ and $\hat \theta_R(s)$ as the angles from $\vec T$ to $\vechat R_L(s)$ and $\vechat R_R(s)$ in the right screw directions of $\vec P_L(s)$ and $\vec P_R(s)$, respectively.
In other words,
\begin{align}
\vechat R_i(s) 
&=  \cos \hat\theta_i(s)\vec T(s) +  \sin \hat\theta_i(s)\left(\vec P_i(s) \times \vec T(s)\right)\nonumber\\
&= \cos \hat\theta_i(s)\vec T(s) +  \sin \hat\theta_i(s) \cos \frac{\sigma_i\rho}{2} \vechat N + \sin \hat\theta_i(s) \sin \frac{\sigma_i\rho}{2} \vechat B,
\label{R_i}
\end{align}
where $i \in \{L, R\}$, $\sigma_R = 1$, and $\sigma_L =  -1$.

The ruling angles can be computed as follows.
By the orthogonality of rulings to surface normals, $\vec R \cdot \vec P_i = 0$, we get its derivative also being zero, i.e., $\vechat R' \cdot \vec P_i +  \vec R \cdot \vec P_i' = 0$.
By the developability of the surface, $\vechat R' \cdot \vec P_i = 0$. Thus
\begin{align}
\vechat R \cdot \vec P_i' = 0.
\end{align}
Using
\begin{align}
\vec P_i &= \cos \frac{\rho}{2} \vechat B + \sigma_i \sin\frac{\rho}{2} \vechat N,
\label{P_i}
\end{align}
we obtain
\begin{align}
\vechat R \cdot \vec P_i' & = \hat K\sin\frac{\sigma_i\rho}{2} \cos \hat\theta_i - \left(\tau + \frac{\sigma_i\rho'}{2}\right) \sin \hat\theta_i.
\end{align}
Using~(\ref{eq:curvature}), we get
\begin{align}
\label{eq:ruling}
 \cot \hat\theta_i = \frac{1}{\hat k}\left(\tau + \frac{\sigma_i\rho'}{2}\right) \cot\frac{\sigma_i\rho}{2}.
\end{align}
Here, we used that the crease is curved, i.e., $\hat k \neq 0$, and properly and smoothly folded, i.e., $\rho \notin \{0, \pi \}$, and thus the crease has no rulings tangent to the crease, i.e., $\theta_i \neq 0$.
This equation is equivalent to Equation (1) of~\cite{Fuchs-Tabachnikov-1999}.
Therefore, a single crease has left and right side rulings that satisfy
\begin{align}
\label{eq:ruling-rule}
 \frac{\cot \hat\theta_L + \cot \hat\theta_R}{2} = \frac{1}{\hat k}\frac{\rho'}{2}  \cot\frac{\rho}{2}.
\end{align}

\paragraph{Ruling angles for special cases.}
In the curved crease designs of the third author,
two special cases are often used:
\begin{enumerate}
\item When the fold angle $\rho$ is constant along the crease, Equation~\ref{eq:ruling-rule} gives $\cot \hat\theta_L +  \cot \hat\theta_R = 0$, i.e., $\hat\theta_L = \pi -\hat\theta_R$, meaning that the rulings \emph{reflect} through the crease.
The angles are given by
\begin{align}
\label{eq:ruling-rule-reflect}
-\cot \hat\theta_L =  \cot \hat\theta_R =  \frac{1}{\hat k}\tau \cot\frac{\rho}{2}.
\end{align}
\item When the folded curve $\vec X$ is a planar curve, i.e., $\tau = 0$, Equation~\ref{eq:ruling-rule} gives
\begin{align}
 \cot \hat\theta_i = \frac{1}{\hat k}\frac{\rho'}{2} \cot\frac{\rho}{2}.
\end{align}
So, in particular, $\hat\theta_L = \hat\theta_R$, meaning that the rulings just penetrate the crease without changing their angle.
\end{enumerate}

\section{Compatibility Between Creases} 
\label{sec:between-creases}
\begin{figure}
  \centering
  \includegraphics[width=0.70\linewidth]{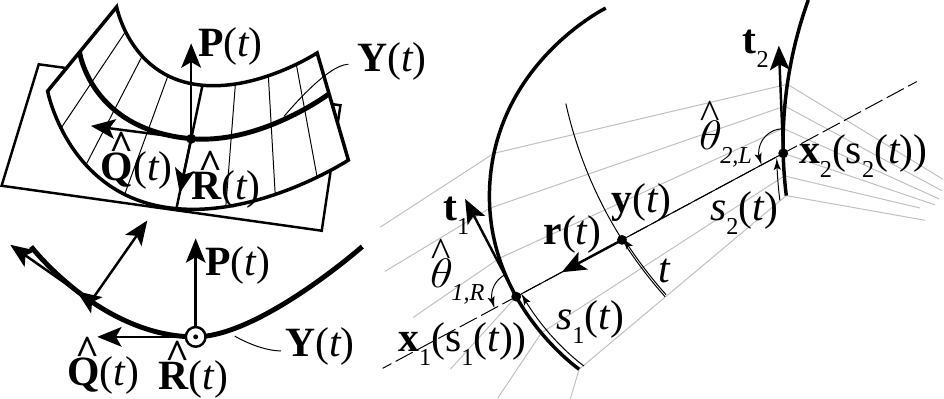}

  \caption{Notation for rule segments: $\left( \vechat Q, \vechat R, \vec P\right)$ frame, and a family of rule segments connecting two curved creases.}
 \label{fig:rulelineMV}
\end{figure}
Next, we give compatibility conditions between creases connected by rule segments.
Consider two curves $\vec X_1(s_1)$ and $\vec X_2(s_2)$, each parameterized by its own arc length,
with a family of rule segments connecting corresponding points of both creases;
refer to Figure~\ref{fig:rulelineMV}.
Suppose the correspondence between $\vec X_1(s_1)$ and $\vec X_2(s_2)$ is given
by two functions $s_1(t)$ and $s_2(t)$ in a single parameter~$t$.
For a sufficiently small patch of rule segments, we can find an
arc length-parameterized principal curvature line $\vec Y(t)$ on a ruled surface
such that each rule segment $\vec X_1(s_1(t))\vec X_2(s_2(t))$
intersects $\vec Y(t)$.
Along the principal curvature line $\vec Y(t)$, we consider the following frame $\left(\vechat Q(t), \vechat R(t), \vec P(t)\right)$, where $\vechat Q(t) := \vechat R(t) \times \vec P(t)$ is the tangent vector of~$\vec Y(t)$.\footnote{The frame is not a Frenet frame but is a Darboux frame.}

The compatibility of surfaces can be sufficiently guaranteed by having the common principal curvature $V(t):={d \vechat Q(t) \over dt}\cdot \vec P(t)$ of the surface at $\vec Y(t)$.
This is because the curve and rulings are intrinsically compatible, and the other principal curvature is $0$ and common.
Because the surface between two creases is a developable surface, the surface should have the common surface normals share the rulings, $\vec P(t) = \vec P_1(s_1(t)) = \vec P_2(s_2(t))$.
$V(t)$ can be computed using $\vec P_i$ for $i \in \{1,2\}$ as (see \cite[Lemma~23 proof]{HuffmanLens_Origami6}
for a detailed derivation):
\begin{align}
V(t) &= {d \vechat Q(t) \over dt}\cdot \vec P(t)\\ 
 	&= {ds_i \over dt} {1 \over \sin \hat\theta_i}\hat K(s_i){\vechat N(s_i)}\cdot \vec T(s_i)\\
 	&= {ds_i \over dt} {1 \over \sin \hat\theta_i}\hat k(s_i)\sigma_i\tan{\rho_i \over 2},
\end{align}
where $\sigma_1= \sigma_R = +1$, and $\sigma_2 = \sigma_L=-1$, considering that the ruling connects to the right side of curve $\vec X_1$ and the left side of curve $\vec X_2$.
This gives the relationship between the corresponding points of two curves as follows:
\begin{align}
\label{eq:compatibility}
{ds_1\over dt}{1 \over \sin \hat\theta_1} \hat k_1 \tan {\rho_1 \over 2} &= 
-{ds_2\over dt}{1 \over \sin \hat\theta_2} \hat k_2 \tan {\rho_2 \over 2}.
\end{align}
Note that we can re-parameterize Equation~\ref{eq:compatibility} by a $C^1$ bijection $t \to t^*$, and the equation stays the same.
This means that we can compute the compatibility using Equation~\ref{eq:compatibility} for an arbitrary $C^1$ bijective parameter $t^*$ along a curve strictly intersecting the rule segments (not necessarily the principal curvature line).
This expression is equivalent to Equation~(19) of~\cite{Tachi-2013-transformables}.

\paragraph{Fold-angle assignment.}
Now observe that ruling angles $\hat\theta_i(s)$, 2D curvature $\hat k_i(s)$, and the correspondence ${ds_i \over ds_j}$ between curves $i$ and $j$ are intrinsic parameters fixed if the crease--rule pattern is given.
The unknown is $\rho(s)$ for each crease, which we call the \emph{fold-angle assignment}; the fold-angle assignment should satisfy Equations~\ref{eq:ruling-rule} and~\ref{eq:compatibility}.
These conditions are complete for every \emph{vertex-free} crease pattern, i.e., a crease pattern without vertices on the strict interior of the paper, on a hole-free (disk-topology) paper. 

\begin{theorem} \label{vertex-free}
A vertex-free uniquely ruled curved crease--rule pattern on a hole-free intrinsically flat piece of paper folds if and only if there exists a fold-angle assignment $\rho(s)$ for every crease such that Equation~\ref{eq:ruling-rule} is satisfied for each crease and Equation~\ref{eq:compatibility} is satisfied for each rule segment between crease points.
\end{theorem}
\begin{proof}
Necessity is as described above.

To prove sufficiency, we describe the folding of many overlapping patches
of the pattern, and show that they agree on their overlap, and thus combine
together compatibly.
First, consider the \emph{crease graph} whose vertices are curved creases,
with two vertices connected by an edge when there are rule segments connecting
the corresponding curved creases.
Further, add ``one-sided edges'' to this graph corresponding to rule segments
that start at the crease and go to the paper boundary.
By the vertex-free and hole-free assumption, the crease graph is a tree.
Hence, if we can show how to fold each crease and its neighboring rule lines,
and show agreement along each edge of the crease graph, then there is a unique
way to combine them together (with no closure constraints to check).

The Fundamental Theorem of Space Curves uniquely determines each folded crease
as a space curve up to rigid motion, provided we can determine the (signed)
curvature $\hat K(s)$ and torsion $\tau(s)$.
We know the crease pattern (which determines the 2D frame
$\left(\vec t(s), \vechat n(s), \vechat b(s)\right)$ of $\vec x(s)$ and its
signed curvature $\hat k(s)$), the ruling (which determines the intrinsic angle
$\hat \theta_i(s)$), and the fold-angle assignment $\rho(s)$.
Equation~\ref{eq:curvature} gives us $\hat K(s)$ from $\hat k(s)$ and $\rho(s)$.
Equation~\ref{eq:ruling} gives us $\tau(s)$ from $\hat \theta_i$ and $\rho(s)$.
Thus we obtain the space curve $\vec X(s)$ up to rigid motion.
From its frame $\left(\vec T(s), \vechat N(s), \vechat B(s)\right)$,
we further obtain $\vechat R_i(s)$ by Equation~\ref{R_i} and
$\vec P_i(s)$ by Equation~\ref{P_i}.
By placing each ruling vector $\vechat R_i(s)$ as a segment starting at
$\vec X(s)$ and whose length equals the corresponding rule line in the 2D
crease--rule pattern, we sweep the two ruled surfaces incident to the crease.

Now consider two creases that share a family of rule segments.
We will prove that the reconstructed ruled surfaces from either crease
agree (up to rigid motion), and thus we can paste together the reconstructions.
First we cover the shared family of rule segments by multiple overlapping
sufficiently small patches such that, for each patch, we can draw the
principal curvature line $\vec Y(t)$ within the patch
(i.e., without getting clipped by the endpoints of the rule segments).
In this way, we can coordinatize the patch as viewed from either crease.
Equation~\ref{eq:compatibility} guarantees that these two coordinatizations
are identical.  Then we can glue together the overlapping patches
to form a unique joining of the two creases by the folded rule segment family.
\end{proof}

For full completeness, we would need to add closure constraints around
vertices in the crease pattern and around holes of the paper.
We leave this to future work.

\paragraph{Rigid-ruling folding.} 
If a folding motion of a piece of paper does not change the crease--rule
pattern throughout the motion, we call it a \emph{rigid-ruling folding}. 
In such a motion,  ${ds_i\over dt}{1 \over \sin \theta_i} \hat k_i$ are constant, so by Equation~\ref{eq:compatibility}, the tangent of half the fold angle of corresponding points keep their proportions to each other. 
In the case of constant-fold-angle creases with reflecting rule segments,
the tangent of half the fold angle at every point is proportional to each other.
\begin{theorem}
If a vertex-free uniquely ruled curved crease--rule pattern with reflecting rulings on hole-free intrinsically flat paper has a properly folded state,
then it has a rigid-ruling folding motion.
\end{theorem}
\begin{proof}
Theorem~\ref{vertex-free} tells us that it suffices to show a continuous
changing of the fold angles $\rho$ while satisfying
Equations~\ref{eq:ruling-rule} and~\ref{eq:compatibility}.
Because we are in the reflecting-ruling case, we can replace
Equation~\ref{eq:ruling-rule} with Equation~\ref{eq:ruling-rule-reflect},
which we can always satisfy by setting $\tau$ accordingly with~$\rho$.
Given one solution to Equation~\ref{eq:compatibility},
denoted $\rho^*_i$ for each crease~$i$,
we can construct a continuous family of solutions by
\begin{align}
\tan{\rho_i \over 2} &= u \tan {\rho^*_i \over 2},
\end{align}
where $u$ ranges continuously from $0$ (completely unfolded state) to~$1$ (target state)
(corresponding to folding time).
Thus we obtain a rigid-ruling folding motion.
\end{proof}

This behavior of ``folded state implies rigid folding motion'' is analogous
to that of flat-foldable quadrilateral meshes \cite{Tachi-2008-quad}.

\section{Naturally Ruled Conic Curved Creases}
Next, we apply the general compatibility conditions from
Section~\ref{sec:between-creases} to the special case of conic curved creases
satisfying the ``naturally ruled'' conditions from
Section~\ref{sec:introduction}: the rule segments reflect at the creases, and
on either side, converge to a focus of the conic
(viewing parabolas as having a second focus at infinity,
which leads to parallel rule segments).
Section~\ref{subsec:finite-focus} covers the finite case, and Section~\ref{subsec:infinite-focus} covers the infinite case.

\subsection{Conic curves sharing a focus of finite distance}
\label{subsec:finite-focus}
Consider conic curves sharing a finitely distant focus; assume that it is at the origin without loss of generality.
We consider the common parameter $t$ moving perpendicular to common radial rulings, i.e., the principal curvature line $\vec Y(t)= (\cos t, \sin t)$ is a circular arc around the origin.
Then the polar coordinates $\left(r(t),t \right)$ of a conic curve is given by
\begin{align}
r(t) = {a\over 1 + e \cos(t-\delta)},
\end{align}
where $e\in \R$ is the \emph{signed eccentricity} of the conic curve, $\delta\in(-{\pi\over 2},{\pi\over 2}]$ is the rotational offset of the whole pattern, and $a\in \R$ describes the scaling as the distance to the curve at $t = \frac{\pi}{2} - \delta$.
The absolute value of $e$ is the eccentricity, while its sign represents whether the closest vertex is on the right side ($+$) or on the left side ($-$).
As flipping the sign of $e$ rotates the curve by $\pi$, the range of the angular offset $\delta \in(-{\pi\over 2},{\pi\over 2}]$ is sufficient to represent all possible alignments of conic curves.
\label{sec:conics}

\begin{wrapfigure}{r}{50mm}
  \centering
  \vspace*{-2ex}
  \includegraphics[width=\linewidth, page=1]{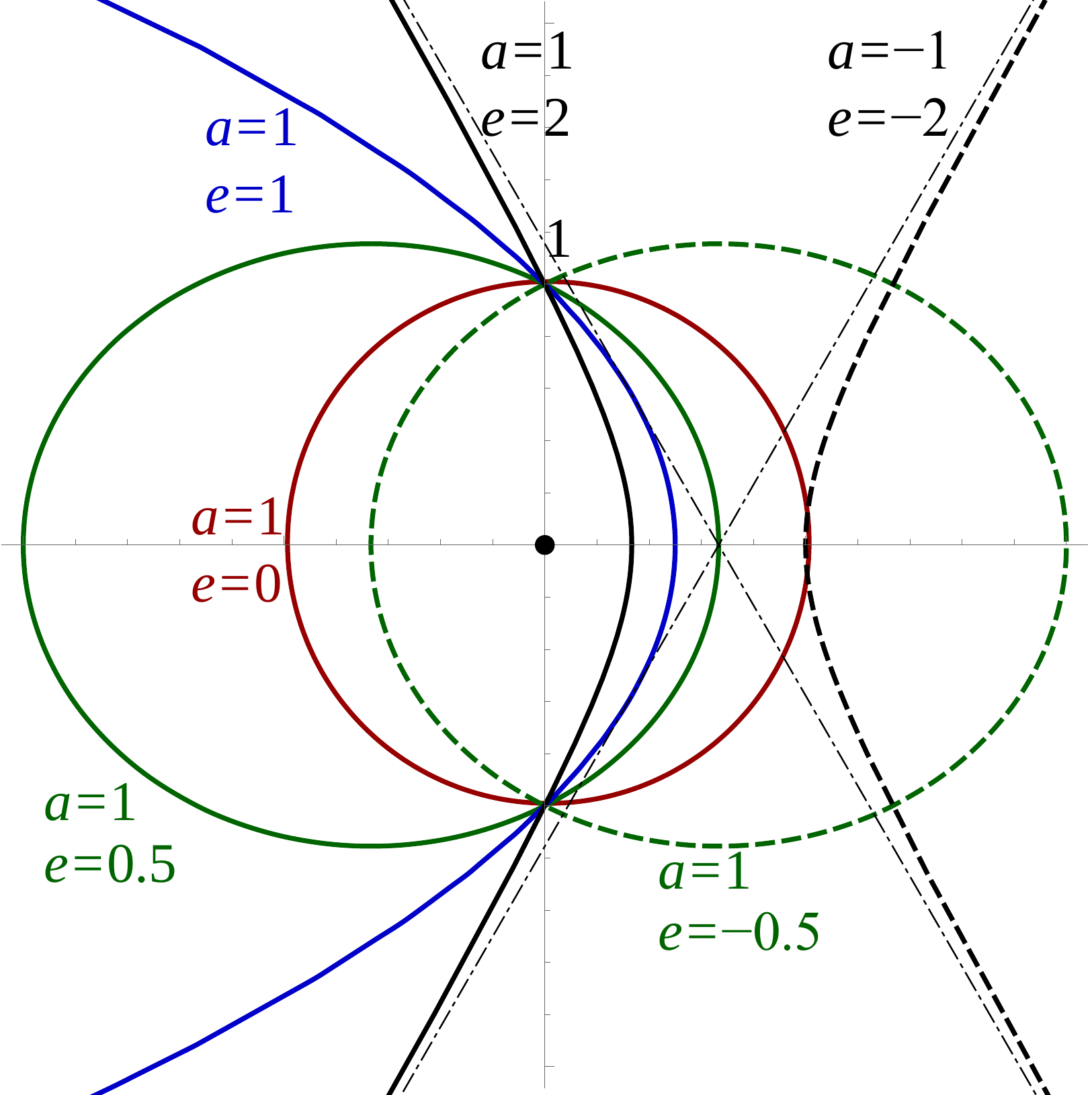}
  \caption{Different conics sharing a focus. All conics use $\delta=0$.}
  \vspace*{-2ex}
 \label{fig:conics}
\end{wrapfigure}

Also, to make sure that rule segments between two curves do not intersect, i.e., crossing over the focus, we forbid $r(t)$ from being negative; so we only take the part of the curves $r(t)>0$.
This restriction to $r(t)>0$ still allows us to represent every conic curve whose focus is at the origin because there is a freedom in choosing both signs of $a$ and~$e$.
More precisely, in the case of an ellipse or parabola, $a>0$ draws the full curves. 
In the case of a hyperbola, $a>0$ and $e>0$ will draw the branch of the hyperbola closer to the focus, and $a<0$ and $e<0$ will draw the other branch of the same hyperbola, i.e., the farther side from the focus.

Now consider creases $r_1(t)$ and $r_2(t)$ with parameter set $\{e,a,\delta \} = \{e_1, a_1,\delta_1 \},\allowbreak \{e_2,a_2,\delta_2 \}$, respectively. Here, we may rotate the whole figure to assume that $\delta_1 =0$. 
Then the ruling vector between the creases, and the tangent vectors, are given by 
\begin{align}
\label{eq:curvature1}
\vechat r &= -\left(\cos t, \sin t\right)\\
\vec t_i &= \mathrm{sgn}(a_i)
\left(\frac{- e_i \sin \delta_i -\sin  t }{\sqrt{e_i^2+2 e_i \cos (t-\delta_i)+1}},\frac{e_i \cos \delta_i +\cos  t }{\sqrt{e_i^2+2 e_i \cos (t-\delta_i)+1}}\right),
\end{align}
and the left-side normal vectors are given by $\vechat n_i = \vec t_i^\perp$, where $\perp$ denotes $\pi\over 2$ counterclockwise rotation of the original vector.
Notice that $t$ represents the arc length around the unit circle centered at the origin, which is a principal curvature line of the common ruled surface.
Therefore, the principal curvature along the unit circle on the right side of crease $1$ is given by
\begin{align}
V_1(t) &= {ds_1 \over dt} {1 \over \sin \hat\theta_1}\hat k_1(s_1)\tan{\rho_1 \over 2}\\
	&=  {d\vec t_1\over dt}\cdot \vechat n_1 \frac{1}{\vechat n_1(t) \cdot \vec r(t)} \tan {\rho_1 \over 2} \\
	&= \frac{\mathrm{sgn}(a_1)}{\sqrt{e_1^2+ 2 e_1 \cos (t-\delta_1)+1}}\tan {\rho_1 \over 2}.
\end{align}
Similarly, the principal curvature on the left side of crease $2$ is given by 
\begin{align}
\label{eq:curvature2}
V_2(t) &= -\frac{\mathrm{sgn}(a_2)}{\sqrt{e_2^2+ 2 e_2 \cos (t-\delta_2)+1}} \tan {\rho_2 \over 2}.
\end{align}
The two creases are compatible if and only if we can find $\rho_1$ and $\rho_2$ such that $V_1(t)\equiv V_2(t)$.
Because we assume reflecting rule lines, the fold angles must be constant, so $\rho_1$ and $\rho_2$ are also constants.

Notice that these expressions do not contain the scale factor $a$ (except for its sign), and thus the compatibility is scale independent.
In particular, there is an obvious solution $\mathrm{sgn}(a_1) = \mathrm{sgn}(a_2)$, $e_1=e_2$, $\delta_1=\delta_2$, and $\rho_1 = -\rho_2$.
\begin{lemma}
\label{lem:scale}
A naturally ruled curved conic crease pattern of two conic curves connected through converging rulings to the shared focus has a valid constant fold-angle assignment if two curves are the scaled version of each other by a positive scale factor with respect to the shared focus.
The fold angles of the creases have the same absolute value and opposite signs.
\end{lemma}
Now we want to narrow down and complete the possible set of corresponding conic curves.
Define the \emph{speed coefficient} of crease $1$ with respect to crease $2$ to be the constant $ p:= \tan {\rho_1 \over 2}   / \tan {\rho_2 \over 2}\neq 0$.
Then Equations~\ref{eq:curvature1} and~\ref{eq:curvature2} require that
\begin{align}
2p^2 e_2\cos (t-\delta_2) - 2 e_1\cos (t-\delta_1) + \left(p^2(e_2^2+1) - (e_1^2+1)\right) &\equiv 0.
\end{align}

\begin{wrapfigure}{r}{40mm}
  \centering
  \vspace*{-2ex}
  \includegraphics[width=\linewidth, page=1]{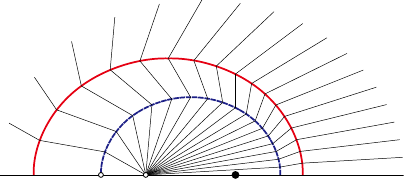}
  \caption{Naturally ruled conic curves that scale to each other (ellipse--ellipse).}
  \vspace*{-2ex}
 \label{fig:scale}
\end{wrapfigure}

If $\delta_1 \neq \delta_2$, then the two harmonic functions differ in their phase, and cannot cancel each other.
So the only possible solution is $e_1 = e_2 = 0$, i.e., two circles scaled with respect to their common center, which falls into the case of Lemma~\ref{lem:scale}.

If $\delta_1 = \delta_2 = 0$, then we get
\begin{align}
2(p^2 e_2 - e_1)\cos t + \left(p^2(e_2^2+1) - (e_1^2+1)\right) &\equiv 0.
\end{align}
Therefore, $2(p^2 e_2 - e_1)=0$ and $p^2 (e_2^2+1)- (e_1^2+1)=0$.
This yields either $e_1= e_2=0$; or $p^2 = \frac{e_1}{e_2}$ and $(-e_1+e_2)(e_1-\frac {1}{e_2})=0$.
The former represents two scaled circles dealt with by Lemma~\ref{lem:scale}.
In the latter nontrivial case, the eccentricity of two curves must be equal or reciprocal to each other.
Curves with equal eccentricity are scaled versions of each other, so this type falls into to the case of Lemma~\ref{lem:scale}.
The only interesting case left is when the eccentricities are reciprocal to each other, i.e., ellipse vs.\ hyperbola.

More precisely, a specific direction of the ellipse---$e > 0$ (left) or $e < 0$ (right)---corresponds to the specific branch of the hyperbola---$e > 0$ (close) or $e < 0$ (far), respectively---based on our formalization that avoids possible ruling intersection.
See Figure~\ref{fig:ellipse-hyperbola} for these cases, where (a) represents the case of positive eccentricity and (b) represents the case of negative eccentricity. 
In either case (a) or (b), the order that the ellipse and hyperbola appear can be arranged by changing their scale factor. 
This alters whether the rulings emanate out from both curves (right) or converge to the other foci (left).

\begin{figure}[htbp]
  \centering
  \includegraphics[width=\linewidth]{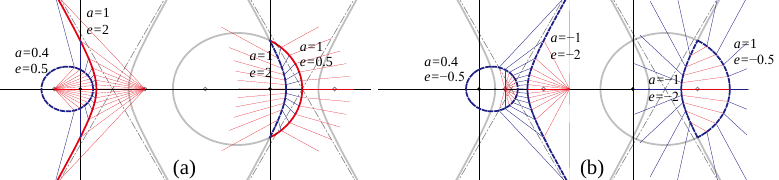}
  \caption{Ellipse--hyperbola interaction. (a) Ellipse and hyperbola compatible with fold-angle assignment of opposite signs. (b) Ellipse and hyperbola compatible with fold-angle assignment of the same sign.}
 \label{fig:ellipse-hyperbola}
\end{figure}

\begin{theorem}
\label{th:reciprocal}
A naturally ruled curved conic crease pattern of two conic curves connected through converging rulings to the shared focus in finite distance has a valid constant fold-angle assignment if and only if
\begin{enumerate}
\item the two curves are scaled versions of each other by a positive scale factor with respect to the shared focus;
\item the two curves are an ellipse and the branch of a hyperbola closer to the focus with reciprocal eccentricities, and the directions from the shared focus to the closest vertex of the ellipse and to the vertex of the hyperbola are the same; or
\item two curves are an ellipse and the branch of a hyperbola farther from the focus with reciprocal eccentricities, and the directions from the shared focus to the farthest vertex of the ellipse and to the vertex of the hyperbola are the same.
\end{enumerate}
In Case 1, the speed coefficient of the two creases is $-1$, i.e., they have the same absolute value but opposite sign;
in Case 2, the speed coefficient of the ellipse with respect to the hyperbola is $-e$ where $e$ is the eccentricity of the ellipse; and
in Case 3, the speed coefficient of the ellipse with respect to the hyperbola is $e$ where $e$ is the eccentricity of the ellipse.
\end{theorem}
\begin{proof}
Necessity follows from the above discussion, and the sufficiency for Case 1 follows from Lemma~\ref{lem:scale}.
So, it suffices to show that Cases 2 and 3 can actually work.
Here, our parameterization of the curves forbid the rule segments from intersecting, and thus there is a valid ruling correspondence between the curves.
Now we check whether the curvature formed from the curves are compatible with each other.

In Case 2, we are matching ellipse 1 with parameters $a_1>0$, $e_1=e$, and hyperbola 2 with $a_2>0$, $e_2=\frac{1}{e}$, where $0<e<1$ (Figure~\ref{fig:ellipse-hyperbola}(a)).
Then
\begin{align}
V_1(t) = \frac{\tan {\rho_1 \over 2}}{\sqrt{e^2+ 2 e \cos t +1}} 
\textrm{\quad and\quad}
V_2(t) = -\frac{e \tan {\rho_2 \over 2}}{\sqrt{e^2+ 2 e \cos t +1}}.
\end{align}
So, the speed coefficient of the ellipse with respect to the hyperbola $ \tan {\rho_1 \over 2}   / \tan {\rho_2 \over 2} = -e$ gives a valid fold-angle assignment.

In Case 3, we are matching ellipse 1 with parameters $a_1>0$, $e_1=-e$, and hyperbola 2 with $a_2<0$, $e_2=-\frac{1}{e}$, where $0<e<1$ (Figure~\ref{fig:ellipse-hyperbola}(b)).
Then 
\begin{align}
V_1(t) = \frac{\tan {\rho_1 \over 2}}{\sqrt{e - 2 e \cos t +1}} 
\textrm{\quad and\quad}
V_2(t) = -\frac{e \tan {\rho_2 \over 2}}{\sqrt{e^2 - 2 e \cos t +1}}.
\end{align}
So,  the speed coefficient of the ellipse with respect to the hyperbola $\tan {\rho_1 \over 2}   / \tan {\rho_2 \over 2} = e$ gives a valid fold-angle assignment.
We can also flip the left and right sides of the curves, and $V_1(t)\equiv V_2(t)$ still holds.
\end{proof}

\subsection{Two parabolas sharing a focus at infinity}
\label{subsec:infinite-focus}
To complete the relation between two conic curves, we consider the special case where the shared focus is at infinity, i.e., two parabolas share parallel rulings.
In this case, the common parameter $t$ can be taken along a principal curvature line $\vec Y(t) = (\textrm{const}, t)$ perpendicular to the common parallel rulings.
Using this parameter, the parabolas are represented by $\begin{pmatrix}a_i(t+\delta_i)^2 + b_i, & t \end{pmatrix}$, for $i=1,2$.
Then, the ruling vector is given by $\vec r(t) = (-1, 0)$ and tangent vector is given by $\vec t = \left(\frac{2 a_i (t+\delta_i)}{\sqrt{4 a_i^2 (t + \delta_i)^2+1}},\frac{1}{\sqrt{4 a_i^2 (t+\delta_i)^2+1}}\right)$.
The principal curvature is computed as
\begin{align}
V_1(t) = -\frac{2 a_1 \tan \frac{\rho_1}{2}}{\sqrt{4 a_1^2 (t+\delta_1)^2+1}}
\textrm{\quad and\quad}
V_2(t) = \frac{2 a_2 \tan \frac{\rho_2}{2}}{\sqrt{4 a_2^2 (t+\delta_2)^2+1}}.
\end{align}
Similar to the cone case, consider the folding speed coefficients $p = \tan \frac{\rho_1}{2}/\tan\frac{\rho_2}{2}$.
Then the equivalence of $V_1$ and $V_2$ yields
\begin{align}
a_1^2a_2^2 (p^2-1) = 0
\textrm{,\quad\quad}
a_1^2a_2^2 (p^2\delta_2 - \delta_1)= 0
\textrm{,\quad and\quad}
p^2 a_1^2-a_2^2 = 0.
\end{align}
Therefore, $\delta_1 = \delta_2 = 0$, $p=\pm1$ and $a_1 = \mp a_2$, where the signs correspond, so only the following cases are possible.
\begin{figure}[htbp]
  \centering
  \includegraphics[width=0.6\linewidth]{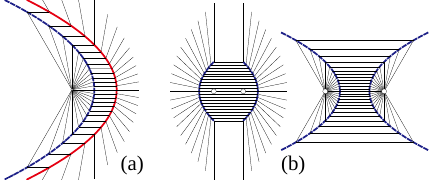}
  \caption{Parabola--parabola interaction with parallel rulings. (a) Parabolas being translation of each other along the ruling direction. (b) Parabolas reflection of each other through a line perpendicular to the rulings.}
 \label{fig:parabola-parabola}
\end{figure}
\begin{theorem}
\label{th:parabola}
Naturally ruled crease pattern of two parabolas connected through parallel rulings has a valid constant fold-angle assignment if and only if
\begin{enumerate}
\item the parabolas are translations of each other in the ruling direction, or
\item the parabolas are mirror reflections of each other with respect to a line perpendicular to the parallel rulings.
\end{enumerate}
In Case 1, two curves have opposite fold angles with the same absolute value; and
in Case 2, two curves have same fold angle.
\end{theorem}
\begin{proof}
Necessity follows the discussions above.
In Case 1 (Figure~\ref{fig:parabola-parabola} Left), $a_1 = a$ and $a_2 = a$, so
\begin{align}
V_1(t) = -\frac{2 a \tan \frac{\rho_1}{2}}{\sqrt{4 a^2 (t)^2+1}}
\textrm{\quad and\quad}
V_2(t) = \frac{2 a \tan \frac{\rho_2}{2}}{\sqrt{4 a^2 (t)^2+1}}.
\end{align}
Therefore, $\tan \frac{\rho_1}{2} = -\tan \frac{\rho_2}{2}$ gives a valid fold-angle assignment.
In Case 2  (Figure~\ref{fig:parabola-parabola} Right), $a_1 = a$ and $a_2 = -a$, so
\begin{align}
V_1(t) = -\frac{2 a \tan \frac{\rho_1}{2}}{\sqrt{4 a^2 (t)^2+1}}
\textrm{\quad and\quad}
V_2(t) = -\frac{2 a \tan \frac{\rho_2}{2}}{\sqrt{4 a^2 (t)^2+1}}.
\end{align}
Therefore, $\tan \frac{\rho_1}{2} = \tan \frac{\rho_2}{2}$ gives a valid fold-angle assignment.
\end{proof}
Theorems~\ref{th:reciprocal} and~\ref{th:parabola} complete the possible cases of naturally ruled crease pattern of two conic curves with valid fold-angle assignment.
To use this result for analyzing conic curved foldings, refer to Table~\ref{table:compatible}.

\section{Analysis of Specific Models}
\label{sec:models}

In this section, we apply the tools of Section~\ref{sec:conics} to
analyze a few models designed by the third author
(before his death in 1999), as documented in \cite{DuksThesis}.
This analysis allows us to detect whether conic curved crease patterns
cannot properly fold with the natural ruling.
When we find that the design satisfies the necessary conditions,
it tells us that things work locally between pairs of creases,
but we remain uncertain whether the full design ``exists''
(can be properly folded) with the natural ruling.
When we find that the design violates a necessary condition,
it does not tell us that the curved crease pattern is impossible to fold,
only that any proper folding must use a different ruling.
We can only conjecture that the third author intended to use the natural
ruling, but he may also have been aware in these cases that the
natural ruling failed to fold.

\paragraph{Huffman tower.}
The classic Huffman tower of Figure~\ref{fig:Huffman tower} cannot fold with
the (drawn) natural ruling because there is an incompatibility between
a circle and parabola, whose eccentricities are $e=0$ and $e=1$ respectively.
On the other hand, we were able to design a simpler tower, shown in
Figure~\ref{fig:Duks tower}, that does satisfy the natural ruling conditions,
as it uses only parabolas (which all have eccentricity~$1$)
reflected orthogonal to rule segments.

We conjecture that the curved crease pattern of the original Huffman tower
still does fold, just with a different (nonreflecting) ruling.
As evidence toward this conjecture, Figure~\ref{fig:discrete Huffman tower}
shows a properly folded discrete model of the Huffman tower which,
taken to the limit, might give a proper folding of the same crease pattern
(but with a different ruling).

\begin{figure}
  \centering
  \includegraphics[width=1.0\linewidth]{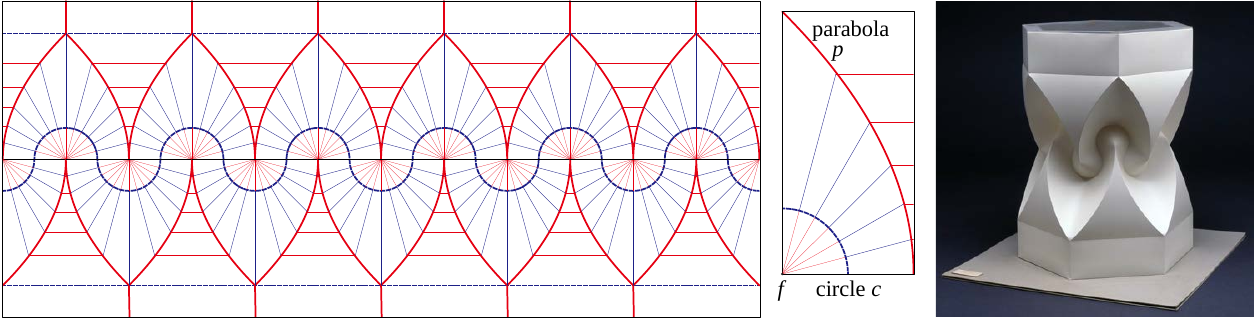}
  \caption{``Hexagonal column with cusps'' designed by the third author \cite{Huffman_Origami5,DuksThesis} drawn with the natural ruling. This crease--rule pattern cannot fold because circle (eccentricity $e=0$) is not compatible with parabola ($e=1$).}
  \label{fig:Huffman tower}
\end{figure}

\begin{figure}
  \centering
  \includegraphics[width=1.0\linewidth]{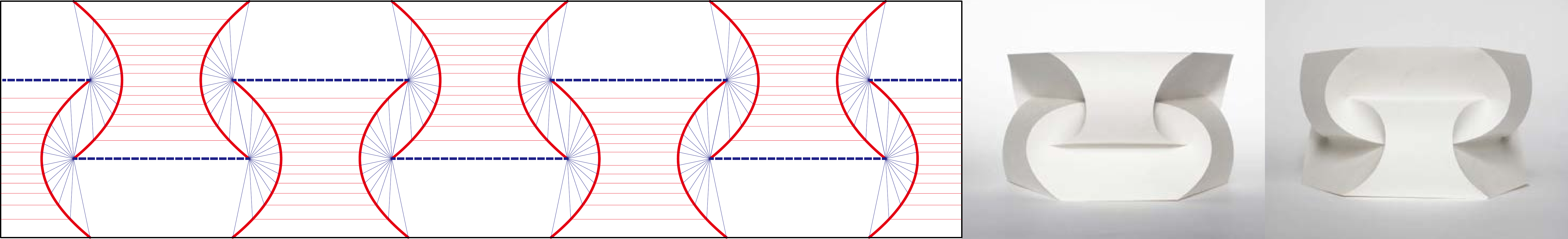}
  \caption{A simple tower using only parabolas, designed by the fourth author
           in the style of the third author, by tiling a part of his ``Arches''
           design \cite[Fig.~2.3.12]{DuksThesis} in a different direction.}
  \label{fig:Duks tower}
\end{figure}

\begin{figure}
  \centering
  \includegraphics[width=1.0\linewidth]{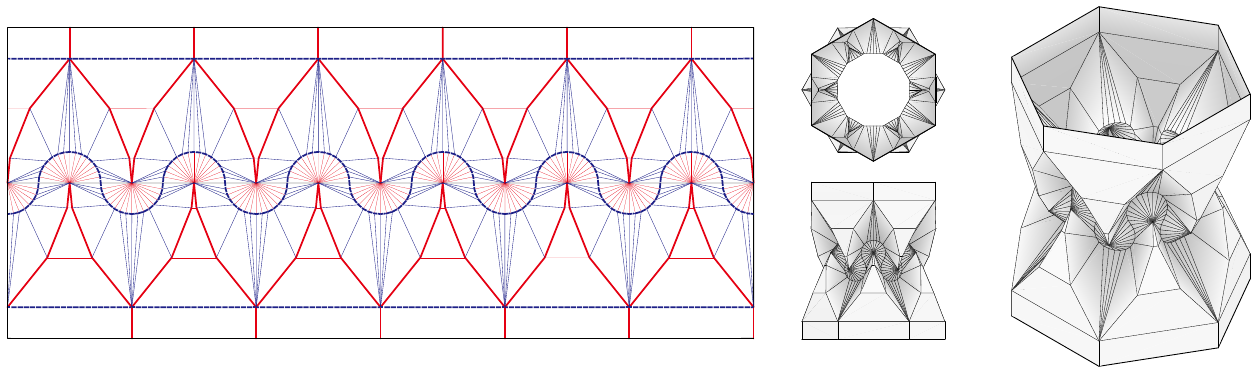}
  \caption{Discrete approximation of Figure~\ref{fig:Huffman tower}, suggesting a possible proper ruling pattern (not reflecting through conics). Produced with the fifth author's software, Freeform Origami.}
  \label{fig:discrete Huffman tower}
\end{figure}

\paragraph{Four columns.}
Figure~\ref{fig:four columns} shows the next design we consider,
which uses two continuous ``pleats'' to form four pipes that seem to
transform from wide to narrow.  The third author exhibited this model
(in a white plexiglass frame) at his exhibition at UCSC in 1978.
The crease pattern consists of quadratic splines made up of parabolic arcs.
Because all parabolas have eccentricity~$1$, all the rule-segment
connections are valid according to our tools.
A hand-drawn sketch on graph paper \cite[Fig.~4.4.71]{DuksThesis}
gives evidence that, in the original design,
parabola $p_2$ is a scaled copy of $p_1$
relative to their shared focus $f_{1,2}$
(and similarly for $p_4$, $p_3$, $f_{3,4}$).
This model indeed properly folds with the natural ruling by Theorem~\ref{vertex-free} as it is a vertex-less curved origami except at the inflection point between parabolas. 
Figure~\ref{fig:four columns} on the right shows its rigid origami simulation.

\begin{figure}
  \centering
  \includegraphics[width=1.0\linewidth]{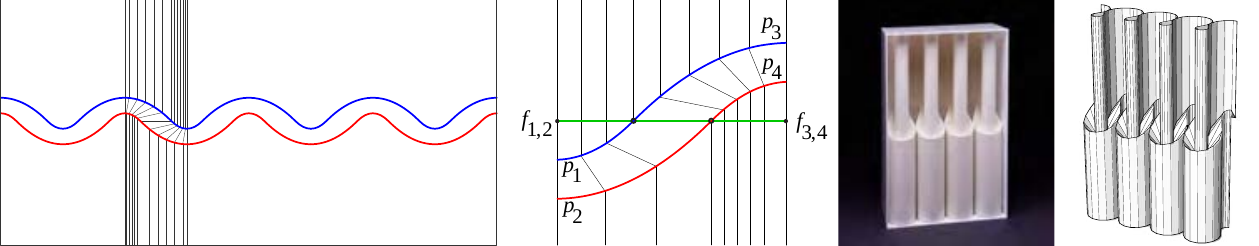}
  \caption{``Four columns'' designed by the third author in 1978 \cite[Fig.~4.4.70]{DuksThesis}.
    Left to right: conjectured crease pattern with natural ruling
    \cite[Fig.~4.4.72]{DuksThesis}; detail of parabolas $p_1,p_2$ with focus
    $f_{1,2}$ and parabolas $p_3,p_4$ with focus $f_{3,4}$;
    photograph by Tony Grant of third author's vinyl model;
    rigid origami simulation produced with Freeform Origami.}
  \label{fig:four columns}
\end{figure}

\paragraph{Angel wings.}
Figure~\ref{fig:angel wings} shows another negative example, which folds
several ``concentric'' hyperbolic pleats into a nearly flat model.
This design was one of the third author's last, completing it one year
before his death; although the third author did not title the piece,
his family calls it ``angel wings''.
The crease pattern consists of entire (half) hyperbolas, which we are
fairly certain share the same foci, and just shift the nearest point
by integral amounts along the vertical axis.
As a consequence, the eccentricities are all different, which means that
the design cannot properly fold with the natural ruling.
Surprisingly, we could still produce a reasonable-looking rigid origami
simulation; we are not sure why this works so well, but it certainly
shows the limits of ``trusting'' a simulation.

\begin{figure}
  \centering
  \includegraphics[width=1.0\linewidth]{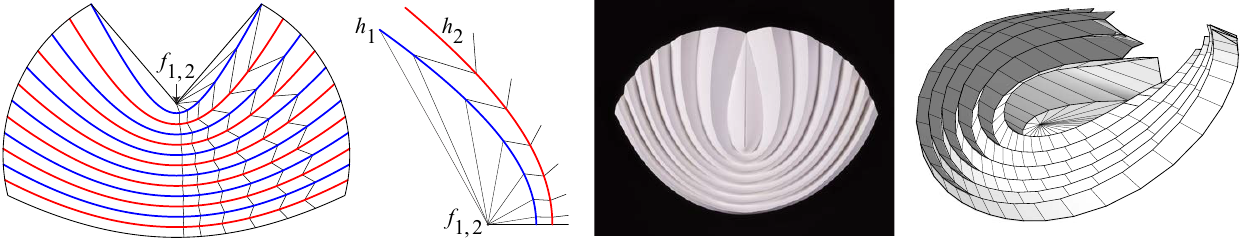}
  \caption{Curved crease design by the third author in 1998 \cite[Fig.~4.4.70]{DuksThesis}.
    Left to right: conjectured crease pattern with natural ruling
    \cite[Fig.~4.5.6]{DuksThesis}; rotated detail of hyperbolas $h_1,h_2$ with
    one focus at $f_{1,2}$ and the other focus at a reflection (not shown);
    photograph by Tony Grant of third author's vinyl model;
    rigid origami simulation produced with Freeform Origami.}
  \label{fig:angel wings}
\end{figure}

\paragraph{Starburst.}
Figure~\ref{fig:starburst} shows a final example, which folds
a rotationally symmetric pattern of three nested levels of closed
curved creases alternating ``bumps in'' and ``bumps out''.
(In fact, a fourth level is drawn, but ended up getting cut into the
paper boundary.)  This model was also exhibited at UCSC in 1978.
The natural ruling here uses lots of cone rulings, but one interaction
we can analyze with our tools: between an outermost hyperbola and the
ellipse immediately within.  The validity of this rule-segment connection
depends on the exact placement of the foci and the resulting eccentricities,
and unfortunately, we do not have precise coordinates for the third author's
design, only a hand-drawn sketch.  But it is definitely possible to construct
a valid interaction.

\begin{figure}
  \centering
  \includegraphics[width=1.0\linewidth]{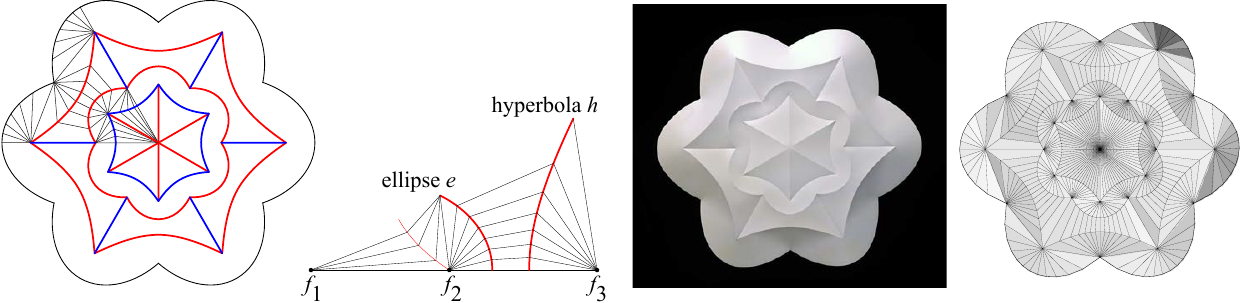}
  \caption{``Starburst'' designed by the third author before 1978
    \cite[Fig.~4.8.23]{DuksThesis}.
    Left to right: conjectured crease pattern with natural ruling
    \cite[Fig.~4.8.24]{DuksThesis}; detail of interaction between outermost
    hyperbola and an ellipse nested within;
    photograph by Tony Grant of third author's vinyl model;
    rigid origami simulation produced with Freeform Origami.}
  \label{fig:starburst}
\end{figure}

\section*{Acknowledgments}

We thank the Huffman family for access to the third author's work,
and permission to continue in his name.

\bibliographystyle{osmebibstyle}
\bibliography{between}

\theaffiliations

\end{document}